\documentclass[runningheads]{llncs}

\usepackage{graphicx}
\usepackage[utf8]{inputenc}
\usepackage{amsmath,amsfonts}
\usepackage{algorithm}
\usepackage{algorithmic}
\usepackage{multirow}
\usepackage{url}
\usepackage{color}
\usepackage{rotating}

\newenvironment{proof*}[1]
  {\renewcommand{\proofname}{#1}
   \begin{proof}}
  {\end{proof}}

\newcommand{\reals}{\mathbb{R}}

\newcommand{\wVec}{\mathbf{w}}
\newcommand{\xVec}{\mathbf{x}}

\newcommand{\cVec}{\mathbf{c}}

\newcommand{\wall}{w_\mathrm{all}}

\newcommand{\hide}[1]{}

\begin{document}

\title{PATHATTACK: Attacking Shortest Paths \\in Complex Networks}

\author{Benjamin~A.~Miller\inst{1} \and
Zohair~Shafi\inst{1} \and
Wheeler~Ruml\inst{2} \and
Yevgeniy~Vorobeychik\inst{3} \and
Tina~Eliassi-Rad\inst{1}\thanks{Point of contact: eliassi@northeastern.edu}\and
Scott~Alfeld\inst{4}}

\authorrunning{B. A. Miller et al.}

\institute{Northeastern University \and
University of New Hampshire \and
Washington University in St. Louis \and
Amherst College}

\maketitle  

\begin{abstract}
Shortest paths in complex networks play key roles in many applications. Examples include routing packets in a computer network, routing traffic on a transportation network, and inferring semantic distances between concepts on the World Wide Web. An adversary with the capability to perturb the graph might make the shortest path between two nodes route traffic through advantageous portions of the graph (e.g., a toll road he owns). In this paper, we introduce the Force Path Cut problem, in which there is a specific route the adversary wants to promote by removing a minimum number of edges in the graph. We show that Force Path Cut is NP-complete, but also that it can be recast as an instance of the Weighted Set Cover problem, enabling the use of approximation algorithms. The size of the universe for the set cover problem is potentially factorial in the number of nodes. To overcome this hurdle, we propose the \texttt{PATHATTACK} algorithm, which via constraint generation considers only a small subset of paths---at most 5\% of the number of edges in 99\% of our experiments. Across a diverse set of synthetic and real networks, the linear programming formulation of Weighted Set Cover yields the optimal solution in over 98\% of cases. We also demonstrate a time/cost tradeoff using two approximation algorithms and greedy baseline methods. This work provides a foundation for addressing similar problems and expands the area of adversarial graph mining beyond recent work on node classification and embedding.

\end{abstract}

\section{Introduction}
\label{sec:intro}

Across a wide range of applications, finding shortest paths among interconnected entities is an important task. Whether routing traffic on a road network, packets in a computer network, ships in a maritime network, or identifying the ``degrees of separation'' between two actors, locating the shortest path is often key to making efficient use of the interconnected entities. By manipulating the shortest path between two popular entities---e.g., people or locations---those along the altered path could have much to gain from the increased exposure. Countering such behavior is important, and understanding vulnerability to such manipulation is a first step to more robust graph mining.

In this paper, we present the \emph{Force Path Cut} problem in which an adversary wants the shortest path between a source node and a target node in an edge-weighted network to go through a preferred path. The adversary has a fixed budget and achieves this goal by cutting edges, each of which has a cost for removal. We show that this problem is NP-complete via a reduction from the $k$-Terminal Cut problem~\cite{Dahlhaus1994}. To solve the optimization problem associated with Force Path Cut, we recast it as a Weighed Set Cover problem, which allows us to use well-established approximation algorithms to minimize the total edge removal cost. We propose the \texttt{PATHATTACK} algorithm, which combines these algorithms with a constraint generation method to efficiently identify paths to target for removal. While these algorithms only guarantee an approximately optimal solution in general, we find that our algorithm yields the optimal-cost solution in a large majority of our experiments.

The main contributions of the paper are as follows: (1) We formally define Force Path Cut and show that it is NP complete. (2) We demonstrate that approximation algorithms for Weighted Set Cover can be leveraged to solve the Force Path Cut problem. (3) We identify an oracle to judiciously select paths to consider for removal, avoiding the combinatorial explosion inherent in na\"{i}vely enumerating all paths. (4)  We propose the \texttt{PATHATTACK} algorithm, which integrates these elements into an attack strategy. (5)  We summarize the results of over 20,000 experiments on synthetic and real networks, in which \texttt{PATHATTACK} identifies the optimal attack in over 98\% of the time.
 
\section{Problem Statement}
\label{sec:model}
We are given a graph $G=(V, E)$, where the vertex set $V$ is a set of $N$ entities and $E$ is a set of $M$ undirected edges representing the ability to move between the entities. In addition, we have nonnegative edge weights $w:E\rightarrow\reals_{\geq0}$ denoting the expense of traversing  edges.  We take the length of a path in $G$ to be the sum of the weights of its edges; thus a shortest path will have minimum total weight.

We are also given two nodes $s,t\in V$. An adversary has the goal of routing traffic from $s$ to $t$ along a given path $p^*$. This adversary removes edges with full knowledge of \(G\) and \(w\), and each edge has a  cost $c:E\rightarrow\reals_{\geq0}$ of being removed. Given a budget $b$, the adversary's objective is to remove a set of edges $E^\prime\subset E$ such that $\sum_{e\in E^\prime}{c(e)}\leq b$ and $p^*$ is the exclusive shortest path from $s$ to $t$ in the resulting graph $G^\prime=(V, E\setminus E^\prime)$.
We refer to this problem as \emph{Force Path Cut}.

We show that this problem is computationally intractable in general by reducing from the 3-Terminal Cut problem, which is known to be NP-complete~\cite{Dahlhaus1994}. In 3-Terminal Cut, we are given a graph $G=(V, E)$ with weights $w$, a budget $b\geq 0$, and three terminal nodes $s_1,s_2,s_3\in V$, and are asked whether a set of edges can be removed such that (1) the sum of the weights of the removed edges is at most $b$ and (2) $s_1$, $s_2$, and $s_3$ are disconnected in the resulting graph (i.e., there is no path connecting any two terminals). Given that 3-Terminal Cut is NP-complete, we prove the following theorem.
\begin{theorem}
Force Path Cut is NP-complete for undirected graphs.\label{thm:FPC_NPC}
\end{theorem}

Here we provide an intuitive sketch of the proof; the formal proof is included in the supplementary material. 
\begin{proof*}{Proof Sketch}
Suppose we want to solve 3-Terminal Cut for a graph $G=(V, E)$ with weights $w$, where the goal is to find $E^\prime\subset E$ such that the terminals are disconnected in $G^\prime=(V, E\setminus E^\prime)$ and $\sum_{e\in E^\prime}w(e)\leq b$. We first consider the terminal nodes: If any pair of terminals shares an edge, that edge must be included in $E^\prime$ regardless of its weight; the terminals would not be disconnected if this edge remains. Note also that for 3-Terminal Cut, edge weights are edge removal costs; there is no consideration of weights as distances. If we add new edges between the terminals that are costly to both traverse and remove, then forcing one of these new edges to be the shortest path requires removing any other paths between the terminal nodes. This causes the nodes to be disconnected in the original graph. We will use a large weight for this purpose: $\wall=\sum_{e\in E}w(e)$, the sum of all weights in the original graph.

We reduce 3-Terminal Cut to Force Path Cut as follows. Create a new graph $\hat{G}=(V, \hat{E})$, where $\hat{E} = E\cup\{\{s_1, s_2\}, \{s_1, s_3\}, \{s_2, s_3\}\}$,
i.e., $\hat{G}$ is the input graph with edges between the terminals added if they did not already exist. In addition, create new weights $\hat{w}$ where, for some $\epsilon>0$,
$\hat{w}(\{s_1, s_2\}) =  \hat{w}(\{s_2, s_3\}) = \wall+2\epsilon$ and $\hat{w}(\{s_1, s_3\}) =2\wall+3\epsilon$,
and $\hat{w}(e)=w(e)$ for all other edges. Let the edge removal costs in the new graph be equal to the weights, i.e., $\hat{c}(e)=\hat{w}(e)$ for all $e\in \hat{E}$. Finally, let the target path consist only of the edge from $s_1$ to $s_3$, i.e., $s=s_1$, $t=s_3$, and $p^*=(s, t)$.

If we could solve Force Path Cut on $\hat{G}$ with weights $\hat{w}$ and costs $\hat{c}$, it would yield a solution to 3-Terminal Cut. We can assume the budget $b$ is at most $\wall$, since this would allow the trivial solution of removing all edges and any additional budget would be unnecessary. If any edges exist between terminals in the original graph $G$, they must be included in the set of edges to remove, and their weights must be removed from the budget, yielding a new budget $\hat{b}$. Using this new budget for Force Path Cut, we will find a solution $\hat{E}^\prime\subset\hat{E}$ if and only if there is a solution $E^\prime\subset E$ for 3-Terminal Cut. A brief explanation of the reasoning is as follows:
\begin{itemize}
    \item When we solve Force Path Cut, we are forcing an edge with a very large weight to be on the shortest path. If any path from $s_1$ to $s_3$ from the original graph remained, it would be shorter than $(s_1, s_3)$. In addition, if any path from $G$ between $s_1$ and $s_2$ remained, its length would be at most $\wall$, and thus a path from $s_1$ to $s_3$ that included $s_2$ would have length at most $2\wall+2\epsilon$. This would mean $(s_1, s_3)$ is not the shortest path between $s_1$ and $s_3$. A similar argument holds for paths between $s_2$ and $s_3$. Thus, no paths can remain between the terminals if we find a solution for Force Path Cut.
    \item If a solution exists for 3-Terminal Cut in $G$, it will yield the solution for Force Path Cut in $\hat{G}$. Any edge added to the graph to create $\hat{G}$ would be more costly to remove than removing all edges from the original $G$, so none will be removed. With all original paths between terminals removed, the only ones remaining from $s_1$ to $s_3$ are $(s_1, s_3)$ and $(s_1, s_2, s_3)$, the former of which is shortest, thus yielding a solution to Force Path Cut.
\end{itemize}
This procedure is illustrated in Figure~\ref{fig:reduction}. We start with the original graph, which may include edges between the terminals. We add high-weight, high-cost edges between the terminals; edges that have long distances but are effectively impossible to remove within the budget. Forcing the shortest path from $s_1$ to $s_3$ to consist of only the direct edge, therefore, requires any previously existing paths between terminals to be removed, resulting in 3 subgraphs, each of which is connected to the others by a single edge. The edges removed by Force Path Cut plus any edges between the terminals in the original graph (dashed lines in the figure) comprise the edges to be removed in the solution to 3-Terminal Cut. 
\begin{figure}[t]
\includegraphics[width=1.5in]{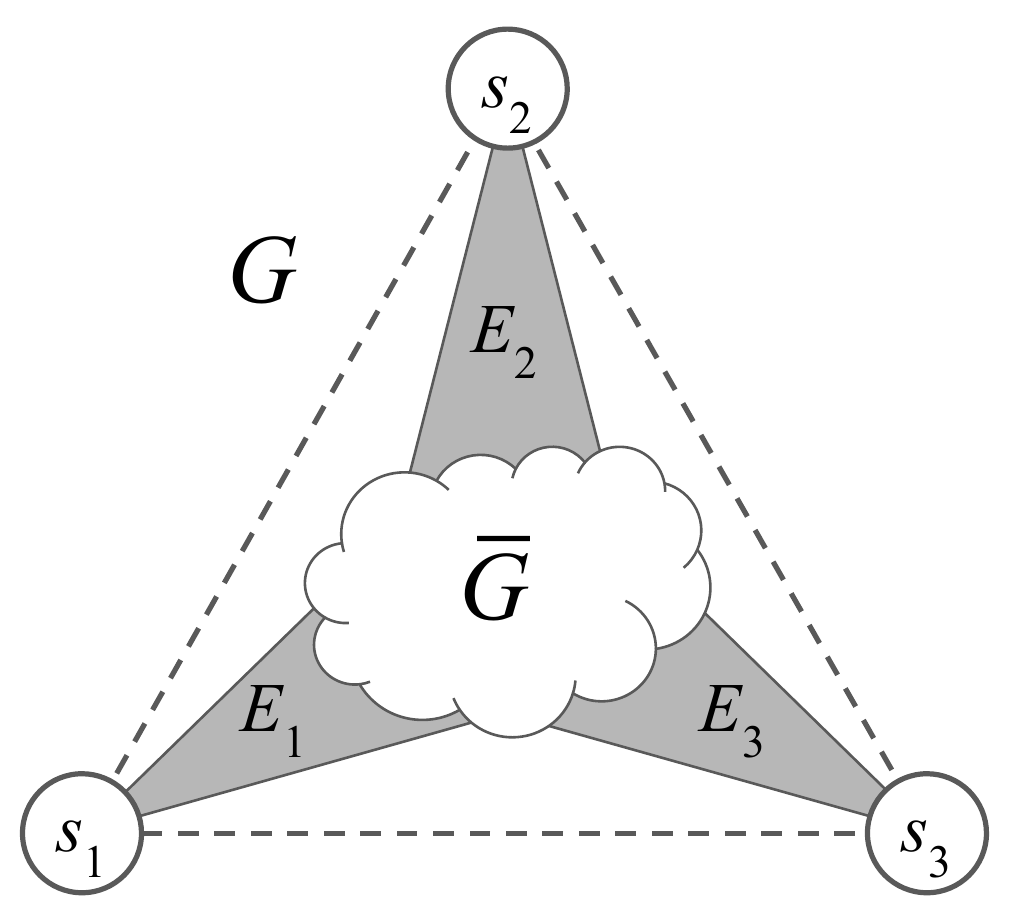}\ \ 
\includegraphics[width=1.5in]{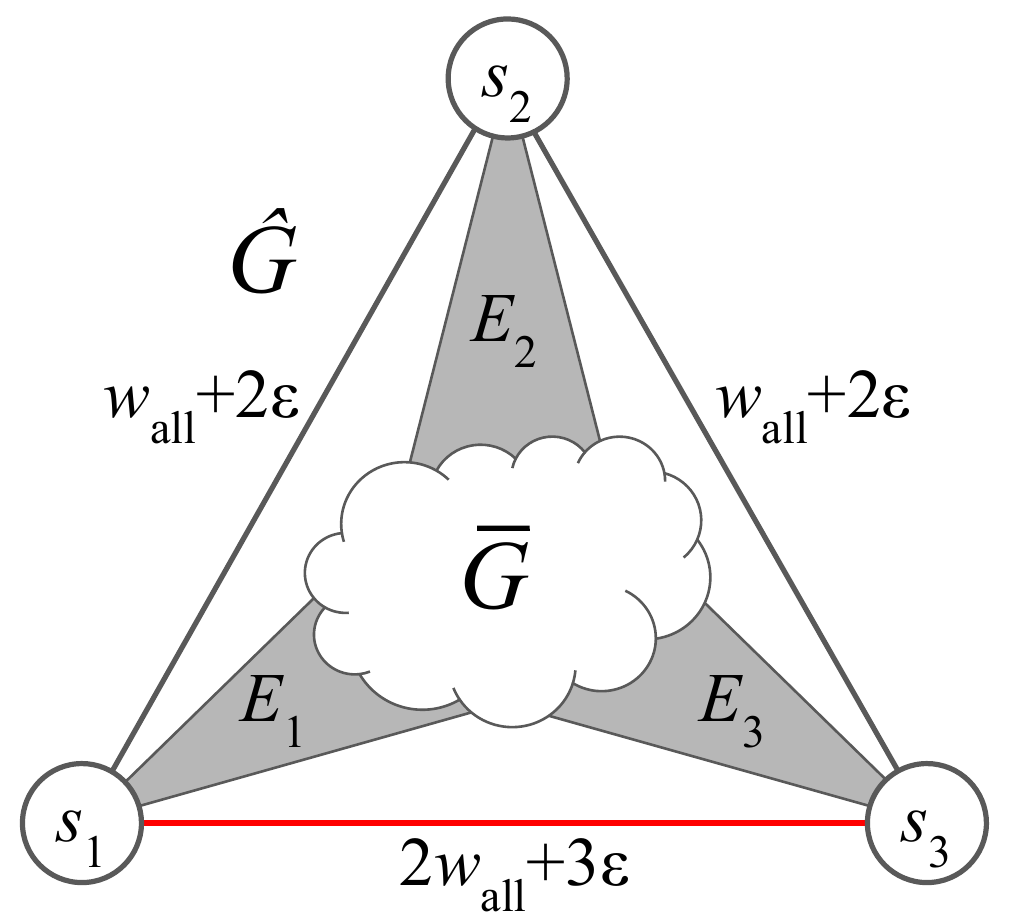}\ \ 
\includegraphics[width=1.5in]{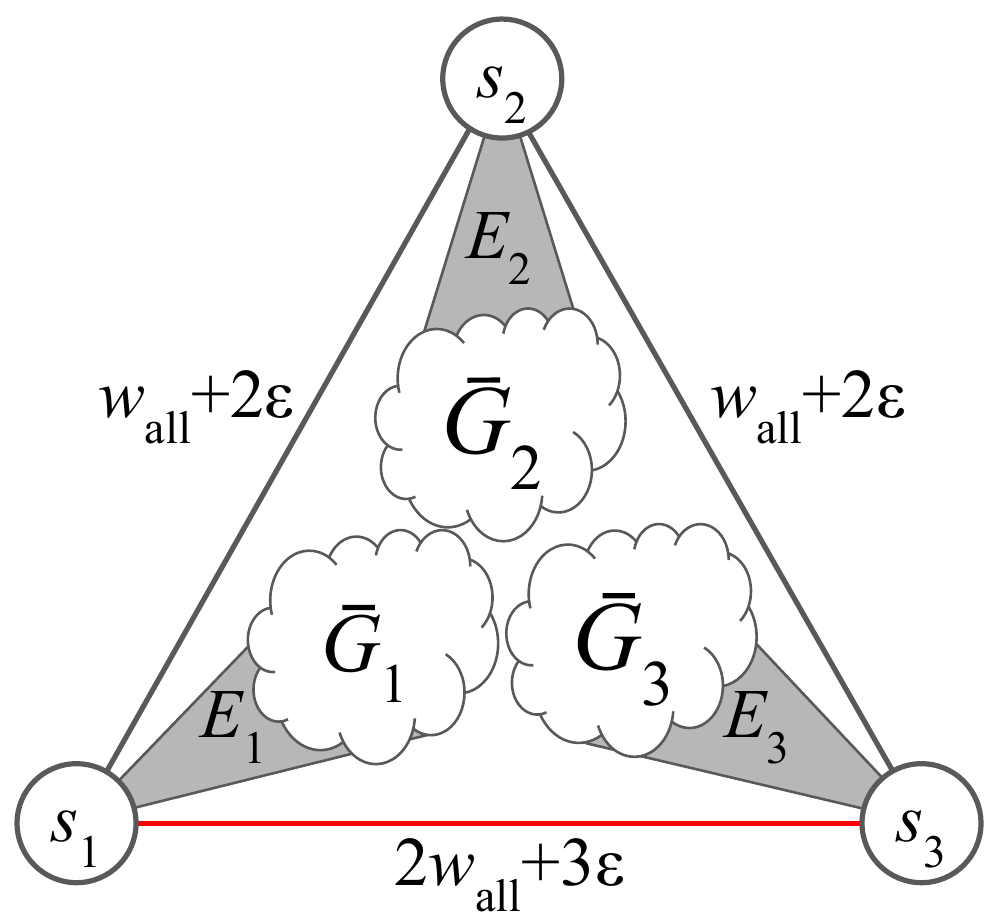}
\caption{Conversion from input to 3-Terminal Cut to Force Path Cut. The initial graph (left) includes 3 terminal nodes $s_1$, $s_2$, and $s_3$, which are connected to the rest of the graph by edges $E_1$, $E_2$, and $E_3$, respectively. The dashed lines indicate the possibility of edges between terminals. The input to Force Path Cut, $\hat{G}$ (center), includes the original graph plus high-weight, high-cost edges between terminals. A single edge comprising $p^*$ is indicated in red. The result of Force Path Cut (right) is that any existing paths between the terminals have been removed, thus disconnecting them in the original graph and solving 3-Terminal Cut.}
\label{fig:reduction}
\end{figure}
\qed
\end{proof*}

\section{Proposed Method: PATHATTACK}
\label{sec:methods}

While finding an optimal solution to Force Path Cut is computationally intractable, we formulate the problem in a way that enables a practical solution that is within a logarithmic factor of optimal.

\subsection{Path Cutting as Set Cover}

The success condition of Force Path Cut is that all paths from $s$ to $t$ aside from $p^*$ must be strictly longer than $p^*$.  This is an example of the (Weighted) Set Cover problem. In Weighted Set Cover, we are given a discrete universe $\mathcal{U}$ and a set of subsets of the universe $\mathcal{S}$, $S\subset\mathcal{U}$ for all $S\in\mathcal{S}$, where each set has a cost $c(S)$. The goal is to choose those subsets whose aggregate cost is minimum yet whose union equals the universe. In Force Path Cut, the elements of the universe to cover are the paths and the sets represent edges: each edge corresponds to a set containing all paths from $s$ to $t$ on which it lies. Including this set in the cover implies removing the edge, thus covering the elements (i.e., cutting the paths). Figure~\ref{fig:translation} shows how Force Path Cut is an example of Weighted Set Cover.

\begin{figure}
    \centering
    \includegraphics[width=\textwidth]{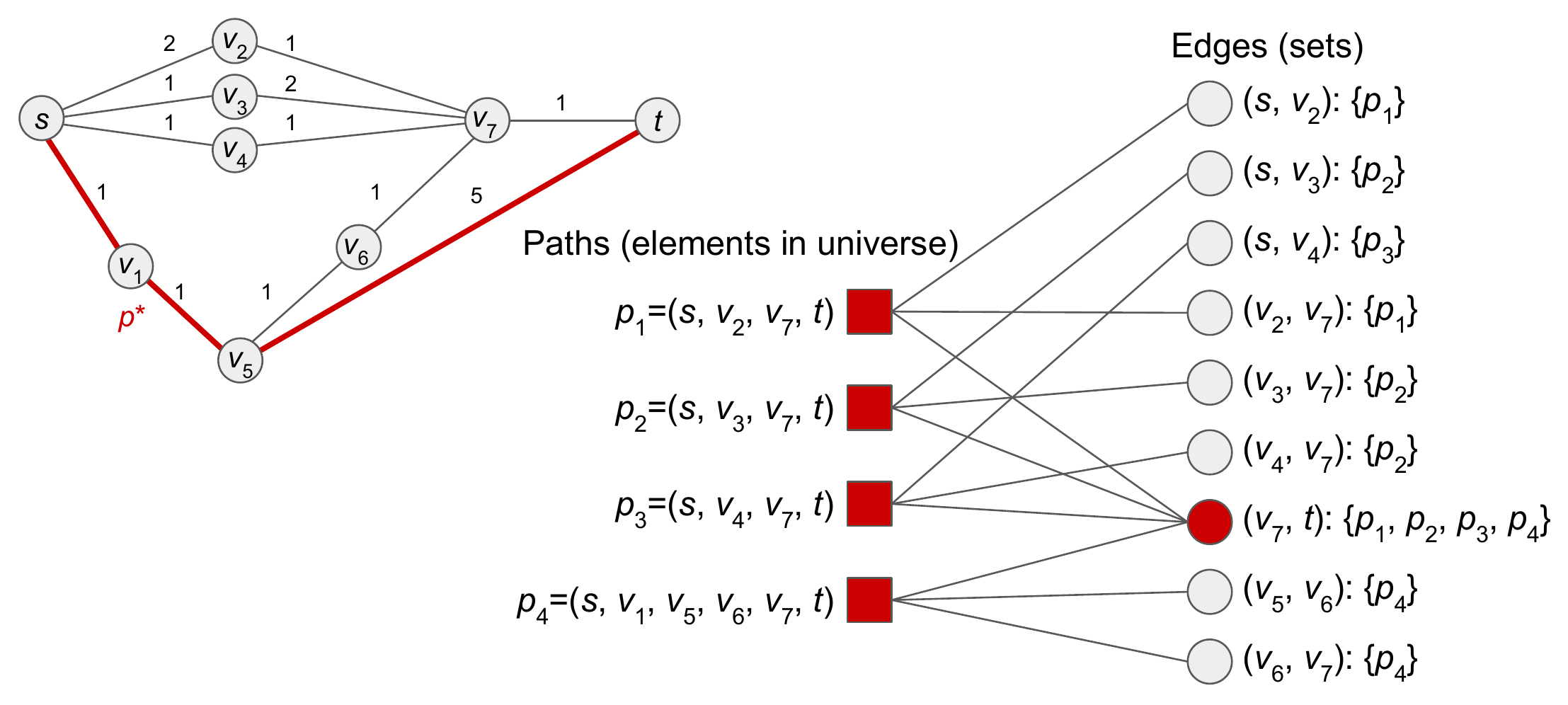}
    \caption{The Force Path Cut problem is an example of the Weighted Set Cover problem. In the bipartite graph on the right, the square nodes represent paths and the circle nodes represent edges. Note that edges along $p^*$ are not included. When the red-colored circle (i.e., edge $(v_7,t)$) is removed, then the red-colored squares (i.e., paths $p_1$, $p_2$, $p_3$, and $p_4$) are removed.}
    \label{fig:translation}
\end{figure}

While Set Cover is NP-complete, there are known approximation algorithms to get a solution within a factor of $O(\log|\mathcal{U}|)$ of the optimal cost. The challenge in our case is that the universe may be extremely large. We address this challenge over the remainder of this section.

\subsection{Linear Programming Formulation}
\label{subsec:LP}

By considering vectors of edges, we formulate an integer program to solve Force Path Cut. Let $\cVec\in\reals_{\geq0}^M$ be a vector of edge costs, where each entry in the vector corresponds to an edge in the graph. We want to minimize the sum of the costs of edges that are cut, which is the dot product of $\cVec$ with a binary vector indicating which edges are cut, denoted by $\Delta\in\{0, 1\}^M$. This means that we optimize over values of $\Delta$ under constraints that (1) $p^*$ is not cut and (2) all other paths from $s$ to $t$ not longer than $p^*$ \emph{are} cut. We represent paths in this formulation by binary indicator vectors---i.e., the vector $\xVec_p\in\{0, 1\}^M$ that represents path $p$ is 1 at entries corresponding to edges in $p$ and 0 elsewhere. Note that this restricts us to considering simple paths---those without cycles---which is sufficient for our purpose. As long as there is one index that is one in both $\Delta$ and $\xVec_p$, the path $p$ is cut. Let $P_p$ be the set of all paths in $G$ from $p$'s source to its destination that are no longer than $p$. The integer linear program formulation of Force Path Cut is as follows: 

\begin{align}
  \hat{\Delta} = &\arg\min_{\Delta} \cVec^\top\Delta\label{eq:minCost}\\
  \text{s.t.} &~\Delta\in\{0, 1\}^{M}\label{eq:binaryCut}\\
  &~\xVec_p^\top\Delta\geq 1~\forall p\in P_{p^*}\setminus \{p^*\}\label{eq:tooShort}\\
  &~\xVec_{p^*}^\top\Delta=0.\label{eq:pStarNotCut}
\end{align}
Here constraint (\ref{eq:tooShort}) ensures that any path shorter than (and thus competing with) $p^*$ will be cut, and constraint (\ref{eq:pStarNotCut}) forbids cutting $p^*$. As mentioned previously, $P_{p^*}$ may be extremely large, which we address in Section~\ref{subsec:constraintGen}. 

The formulation (\ref{eq:minCost})--(\ref{eq:pStarNotCut}) is analogous to the formulation of Set Cover as an integer program~\cite{Vazirani2003}. The goal is to minimize the cost of covering the universe -- i.e., for each element $x\in \mathcal{U}$, at least one set $S\in\mathcal{S}$ where $x\in S$ is included. Letting $\delta_S$ be a binary indicator of the inclusion of subset $S$, the integer program formulation of Set Cover is
\begin{align}
  \hat{\hat{\delta}} = &\arg\min_{\delta} \sum_{S\in\mathcal{S}}{c(S)\delta_S}\label{eq:minSetCost}\\
  \text{s.t.} &~\delta_S\in\{0, 1\}~\forall S\in\mathcal{S}\label{eq:binaryInclusion}\\
  &~\sum_{S\in\{S^\prime\in\mathcal{S}|x\in S\}}{\delta_S}\geq 1~\forall x\in\mathcal{U}\label{eq:coverx}.
\end{align}
Equations (\ref{eq:minCost}), (\ref{eq:binaryCut}), and (\ref{eq:tooShort}) are analogous to (\ref{eq:minSetCost}), (\ref{eq:binaryInclusion}), and (\ref{eq:coverx}), respectively. The constraint (\ref{eq:pStarNotCut}) can be incorporated by not allowing some edges to be cut, which manifests itself as removing some subsets from $\mathcal{S}$.

We formulate Force Path Cut as Set Cover, which enables the use of approximation algorithms to solve the problem. We consider two approximation algorithms. The first method, \texttt{GreedyPathCover}, iteratively adds the most cost-effective subset: that with the highest number of uncovered elements per cost. In Force Path Cut, this is equivalent to iteratively cutting the edge that removes the most paths per cost. The pseudocode is shown in Algorithm~\ref{alg:greedySC}. In this algorithm, we have a fixed set of paths $P$, which is not necessarily \emph{all} paths that must be cut. Note that this algorithm only uses costs, not weights: the paths of interest have already been determined and we only need to determine the cost of breaking them. The \texttt{GreedyPathCover} algorithm performs a constant amount of work at each edge in each path. We use a heap to find the most cost effective edge in constant time. We also use lazy initialization to avoid initializing entries in the tables associated with edges that do not appear in any paths. This means that the entire algorithm runs in time that is linear in the sum of the number of edges across paths. Since a path may visit all nodes in the graph, in the worst case, this leads to a running time of $O(|P|N)$. This algorithm is known to have a worst-case approximation factor of the harmonic function of the size of the universe~\cite{Vazirani2003}. That is,
\begin{equation}
    H_{|\mathcal{U}|}=\sum_{n=1}^{|\mathcal{U}|}{\frac{1}{n}},
\end{equation}
which implies that the \texttt{GreedyPathCover} algorithm has a worst-case approximation factor of $H_{|P_{p^*}|-1}$. We, however, achieve a much tighter bound, which is discussed in Section~\ref{subsec:pathattack}.
\begin{algorithm}
\begin{algorithmic}
\STATE \textbf{Input:} Graph $G=(V, E)$, costs $c$, target path $p^*$, path set $P$
\STATE \textbf{Output:} Set $E^\prime$ of edges to cut
\STATE $T_P\gets$ empty hash table ~ ~ $\langle\langle$set of paths for each edge$\rangle\rangle$
\STATE $T_E\gets$ empty hash table ~ ~ $\langle\langle$set of edges for each path$\rangle\rangle$
\STATE $N_P\gets$ empty hash table ~ ~ $\langle\langle$path count for each edge$\rangle\rangle$
\FORALL{$e \in E$}
\STATE $T_P[e]\gets\emptyset$
\STATE $N_P[e]\gets0$
\ENDFOR
\FORALL{$p\in P$}
\STATE $T_E[p]\gets\emptyset$
\FORALL{edges $e$ in $p$ and not $p^*$}
\STATE $T_P[e]\gets T_P[e]\cup \{p\}$
\STATE $T_E[p]\gets T_E[p]\cup \{e\}$
\STATE $N_P[e]\gets N_P[e]+1$
\ENDFOR
\ENDFOR
\STATE $E^\prime\gets\emptyset$
\WHILE{$\max_{e\in E}{N_P[e]}>0$}
\STATE $e^\prime\gets\arg\max_{e\in E}{N_P[e]/c(e)}$ ~ ~ $\langle\langle$find most cost-effective edge$\rangle\rangle$
\STATE $E^\prime\gets E^\prime\cup\{e^\prime\}$
\FORALL{$p \in T_P[e^\prime]$} 
\FORALL{$e_1\in T_E[p]$}
\STATE $N_P[e_1] \gets N_P[e_1] - 1$ ~ ~ $\langle\langle$decrement path count$\rangle\rangle$
\STATE $T_P[e_1] \gets T_P[e_1] \setminus {p}$ ~ ~ $\langle\langle$remove path$\rangle\rangle$
\ENDFOR
\STATE $T_E[p]\gets\emptyset$ ~ ~ $\langle\langle$clear edges$\rangle\rangle$
\ENDFOR
\ENDWHILE
\RETURN $E^\prime$
\end{algorithmic}
\caption{GreedyPathCover}
\label{alg:greedySC}
\end{algorithm}

The second approximation algorithm we consider involves relaxing the integer constraint into the reals and rounding the resulting solution. We refer to this algorithm as \texttt{LP-PathCover}. In this case, we replace condition (\ref{eq:binaryCut}) with
\begin{equation}
    \Delta\in[0, 1]^{M}
\end{equation}
and get a $\hat{\Delta}$ that may contain non-integer entries. Following the procedure in~\cite{Vazirani2003}, we apply randomized rounding as follows for each edge $e$:
\begin{enumerate}
    \item Treat the corresponding entry $\hat{\Delta}_e$ as a probability.
    \item Draw $\lceil\ln{4|P|}\rceil$ independent Bernoulli random variables with probability $\hat{\Delta}_e$.\label{it:Bernoulli}
    \item Cut $e$ if and only if at least one random variable from step~\ref{it:Bernoulli} is 1.
\end{enumerate}
If the result either does not cut all paths or is too large---i.e., greater than $4\ln{4|P|}$ times the fractional (relaxed) cost---we try again until it works. These conditions are both satisfied with probability greater than 1/2, so the expected number of attempts to get a valid solution is less than 2. By construction, the approximation factor is $4\ln{4|P|}$ in the worst case. The running time is dominated by running the linear program; the remainder of the algorithm is (with high probability) linear in the number of edges and logarithmic in the number of constraints $|P|$. Algorithm~\ref{alg:LPCut} provides the pseudocode for \texttt{LP-PathCover}.
\begin{algorithm}
\begin{algorithmic}
\STATE \textbf{Input:} Graph $G=(V, E)$, costs $\cVec$, path $p^*$, path set $P$
\STATE \textbf{Output:} Binary vector $\Delta$ denoting edges to cut
\STATE $\hat{\Delta}\gets$ relaxed cut solution to (\ref{eq:minCost})--(\ref{eq:tooShort}) with paths $P$
\STATE $\Delta \gets \mathbf{0}$
\STATE $E^\prime\gets\emptyset$
\STATE not\_cut$\gets$\TRUE
\WHILE{$\cVec^\top\Delta > \cVec^\top\hat{\Delta}(4\ln{4|P|})$ \OR not\_cut}
\STATE$E^\prime\gets\emptyset$
\FOR{$i\gets 1$ to $\lceil\ln{4|P|}\rceil$}
\STATE $E_1\gets\{e\in E \textrm{ with probability }\hat{\Delta}_e\}$ ~ ~ $\langle\langle$randomly select edges based on $\hat{\Delta}\rangle\rangle$
\STATE $E^\prime\gets E^\prime\cup E_1$
\ENDFOR
\STATE $\Delta\gets$ indicator vector for $E^\prime$
\STATE not\_cut$\gets(\exists p\in P$ where $p$ has no edge in $E^\prime$)
\ENDWHILE
\RETURN $\Delta$
\end{algorithmic}
\caption{LP-PathCover}
\label{alg:LPCut}
\end{algorithm}

\subsection{Constraint Generation}
\label{subsec:constraintGen}

In general, it is intractable to include  every path from $s$ to $t$. Take the example of an $N$-vertex clique (a.k.a.~complete graph) in which all edges have weight $1$ except the edge from $s$ to $t$, which has weight $N$, and let $p^*=(s, t)$. Since all simple paths other than $p^*$ are shorter than $N$, all of those paths will be included as constraints in (\ref{eq:tooShort}), including $(N-2)!$ paths of length $N-1$. If we only explicitly include constraints corresponding to the two- and three-hop paths (a total of $(N-2)^2+(N-2)$ paths), then the optimal solution will be the same as if we had included all constraints: cut the $N-2$ edges around either $s$ or $t$ that do not directly link $s$ and $t$. Optimizing using only necessary constraints is the other technique we use to make an approximation of Force Path Cut tractable.

Constraint generation is a technique for automatically building a relatively small set of constraints when the total number is extremely large or infinite~\cite{Ben-Ameur2006,Letchford2013}. The method requires an oracle that, given a proposed solution, returns a constraint that is being violated. This constraint is then explicitly incorporated into the optimization, which is run again and a new solution is proposed. This procedure is repeated until the optimization returns a feasible point or determines there is no feasible region.

Given a proposed solution to Force Path Cut---obtained by either approximation algorithm from Section~\ref{subsec:LP}---we have an oracle to identify unsatisfied constraints in polynomial time. We find the shortest path $p$ in $G^\prime=(V, E\setminus E^\prime)$ aside from $p^*$. If $p$ is not longer than $p^*$, an additional constraint is required, it needs to be included in the set of paths $P$. We combine this constraint generation oracle with the approximation algorithms to create our proposed method \texttt{PATHATTACK}.

\subsection{PATHATTACK}
\label{subsec:pathattack}

Combining the above techniques, we propose the \texttt{PATHATTACK} algorithm, which enables flexible computation of attacks to manipulate shortest paths. Starting with an empty set of path constraints, \texttt{PATHATTACK} alternates between finding edges to cut and determining whether removal of these edges results in $p^*$ being the shortest path from $s$ to $t$. Algorithm~\ref{alg:PATHATTACK} provides \texttt{PATHATTACK}'s pseudocode. Depending on time or budget considerations, an adversary can vary the underlying approximation algorithm.
\begin{algorithm}
\begin{algorithmic}
\STATE \textbf{Input:} Graph $G=(V, E)$, cost function $c$, weights $w$, target path $p^*$, method flag $l$
\STATE \textbf{Output:} Set $E^\prime$ of edges to cut
\STATE $E^\prime\gets\emptyset$
\STATE $P\gets\emptyset$
\STATE $\cVec\gets$ vector from costs $c(e)$ for $e\in E$
\STATE $G^\prime\gets(V, E\setminus E^\prime)$
\STATE $s, t\gets$ source and destination nodes of $p^*$
\STATE $p\gets$ shortest path from $s$ to $t$ in $G^\prime$ (not including $p^*$)
\WHILE{$p$ is not longer than $p^*$}
\STATE $P\gets P\cup\{p\}$
\IF{$l$}
\STATE $\Delta\gets$ \texttt{LP-PathCover}$(G, \cVec, p^*, P)$
\STATE $E^\prime\gets$ edges from $\Delta$
\ELSE
\STATE $E^\prime\gets$ \texttt{GreedyPathCover}$(G, c, p^*, P)$
\ENDIF
\STATE $G^\prime\gets(V, E\setminus E^\prime)$
\STATE $p\gets$ shortest path from $s$ to $t$ in $G^\prime$ (not including $p^*$) using weights $w$
\ENDWHILE
\RETURN $E^\prime$
\end{algorithmic}
\caption{PATHATTACK}
\label{alg:PATHATTACK}
\end{algorithm}

We improve the upper bound on the size of the approximation to be better than the harmonic number of the size of the universe (number of paths). The proof of the bound for the greedy Set Cover method uses the fact that, at the $i$-th iteration, the selected set will be at most $\frac{1}{i}$ times the optimal solution. While the number of paths can be astronomically large, the number of sets is bounded by the number of edges. This yields the following proposition.
\begin{proposition}
The approximation factor of \texttt{PATHATTACK-Greedy} is at most $H_M$ times the optimal solution to Force Path Cut.
\end{proposition}

In addition, since we have $M$ variables in the linear program for \texttt{PATHATTACK-LP}, there will be at most $M$ active constraints at its optimal solution. Thus, even if we were to include more than $M$ constraints, we could arrive at the same conclusion using only $M$ of them.  Applying the results of \cite{Vazirani2003} yields:
\begin{proposition}
\texttt{PATHATTACK-LP} yields a worst-case $O(\log M)$ approximation to Force~Path Cut with high probability.
\end{proposition}

\section{Experiments}
\label{sec:setup}

This section presents the baseline methods, the networks used in experiments, the experimental setup, and results. 

\subsection{Baseline Methods}
We consider two simple greedy methods as baselines for assessing performance. Each of these algorithms iteratively computes the shortest path $p$ between $s$ and $t$; if $p$ is not longer than $p^*$, it uses some criterion to cut an edge from $p$. When we cut the edge with minimum cost, we refer to the algorithm as \texttt{GreedyCost}. We also consider a version where we cut the edge in $p$ with the largest ratio of eigenscore\footnote{The eigenscore of an edge is the product of the entries in the principal eigenvector of the adjacency matrix corresponding to the edge's vertices.} to cost, since edges with high eigenscores are known to be important in network  flow~\cite{tongCIKM2012}. This version of the algorithm is called \texttt{GreedyEigenscore}. In both cases, edges from $p^*$ are not allowed to be cut.

\subsection{Synthetic and Real Networks}
\label{sec:graphdata}
Our experiments are on synthetic and real networks. All networks are undirected.

For the synthetic networks, we run five different random graph models to generate 100 synthetic networks of each model. We pick parameters to yield networks with similar numbers of edges ($\approx 160$K). For the Erd\H{o}s--R\'{e}nyi (ER) networks, we use a probability of 0.00125 of forming an edge. For the Barab\'{a}si-Albert (BA) networks, each node enters the network with degree 10. For the Stochastic Kronecker graph (KR) networks, we set the density parameter  to $0.00125$. The Lattice (LAT) networks are on a $285\times 285$ lattice. The Complete (COMP) networks are a 565-node clique.  Table \ref{table:syn_net_prop}, in Appendix~\ref{sec:datasets}, provides a summary characterization of these networks.

For real networks, we have seven weighted and unweighted networks. The unweighted networks are Wikispeedia graph (WIKI), Oregon autonomous system network (AS), and Pennsylvania road network (PA-ROAD). The weighted networks are Central Chilean Power Grid (GRID), Lawrence Berkeley National Laboratory network data (LBL), the Northeast US Road Network (NEUS), and the DBLP coauthorship graph (DBLP).  The networks range from 444 edges on 347 nodes to over 8.3M edges on over 1.8M nodes, with average degree ranging from over 2.5 to over 46.5 nodes and number of triangles ranging from 40 to close to 27M. Appendix~\ref{sec:datasets} (in supplementary material) provides further details on all the real networks, including web locations to download them.

For the synthetic networks and the unweighted real networks, we try three different edge-weight initialization schemes: Poisson, uniformly distributed, or equal weights. For Poisson weights, each edge $e$ has an independently random weight $w_e=1+w_e^\prime$, where $w_e^\prime$ is drawn from a Poisson distribution with rate parameter 20. For uniform weights, each edge's weight is drawn from a discrete uniform distribution of integers from 1 to 41. This yields the same average weight as the Poisson weights. For equal weights, all edges have equal weights.

\subsection{Experimental Setup}
For each graph---considering graphs with different edge-weighting schemes as distinct---we run 100 experiments unless otherwise noted. For each graph, we select $s$ and $t$ uniformly at random among all nodes, with the exception of LAT, PA-ROAD, and NEUS, where we select $s$ uniformly at random and select $t$ at random among nodes 50 hops away from $s$\footnote{This alternative method of selecting the destination was used due to the computational expense of identifying successive shortest paths in large grid-like networks.}. Given $s$ and $t$, we identify the shortest simple paths and use the 100th, 200th, 400th, and 800th shortest as $p^*$ in four experiments. For the large grid-like networks (LAT, PA-ROAD, and NEUS), this procedure is run using only the 60-hop neighborhood of $s$. We focus on the case where the edge removal cost is equal to the weight (distance).

The experiments were run on a Linux cluster with 32 cores and 192 GB of memory per node. The linear program in \texttt{PATHATTACK-LP} was implemented using Gurobi 9.1.1, and shortest paths were computed using  \texttt{shortest\_simple\_paths} in NetworkX.\footnote{Gurobi is at \url{https://www.gurobi.com}. NetworkX is at \url{https://networkx.org}.}

\subsection{Results}
\label{sec:results}
Across over 20,000 experiments, \texttt{PATHATTACK-LP} finds the optimal solution in over 98\% of cases. We can tell the solution is optimal because the relaxed linear program results in integer values.

We treat the result of \texttt{GreedyCost} as our baseline cost and report the cost of other algorithms' solutions as a reduction from the baseline. With one exception\footnote{\texttt{GreedyEigenscore} only outperforms \texttt{GreedyCost} in COMP with uniform random weights.}, \texttt{GreedyCost} outperforms \texttt{GreedyEigenscore} in both running time and edge removal cost, so we omit the \texttt{GreedyEigenscore} results for clarity of presentation. We use the algorithms to minimize cost without explicitly considering a budget. In practice, the adversary would run one of the optimization algorithms, compare budget and cost, and decide whether the attack is possible given resource constraints. 
Figure~\ref{fig:plots_sim_gs} shows the results on synthetic networks,
Figure~\ref{fig:plots_ruw_gs} shows the results on real networks with synthetic edge weights, and
Figure~\ref{fig:plots_rw_gs} shows the results on real weighted networks. In these figures, the 800th shortest path is used as $p^*$; other results were similar and omitted for space.

\begin{figure}[t]
    \centering
    \includegraphics[width=\textwidth]{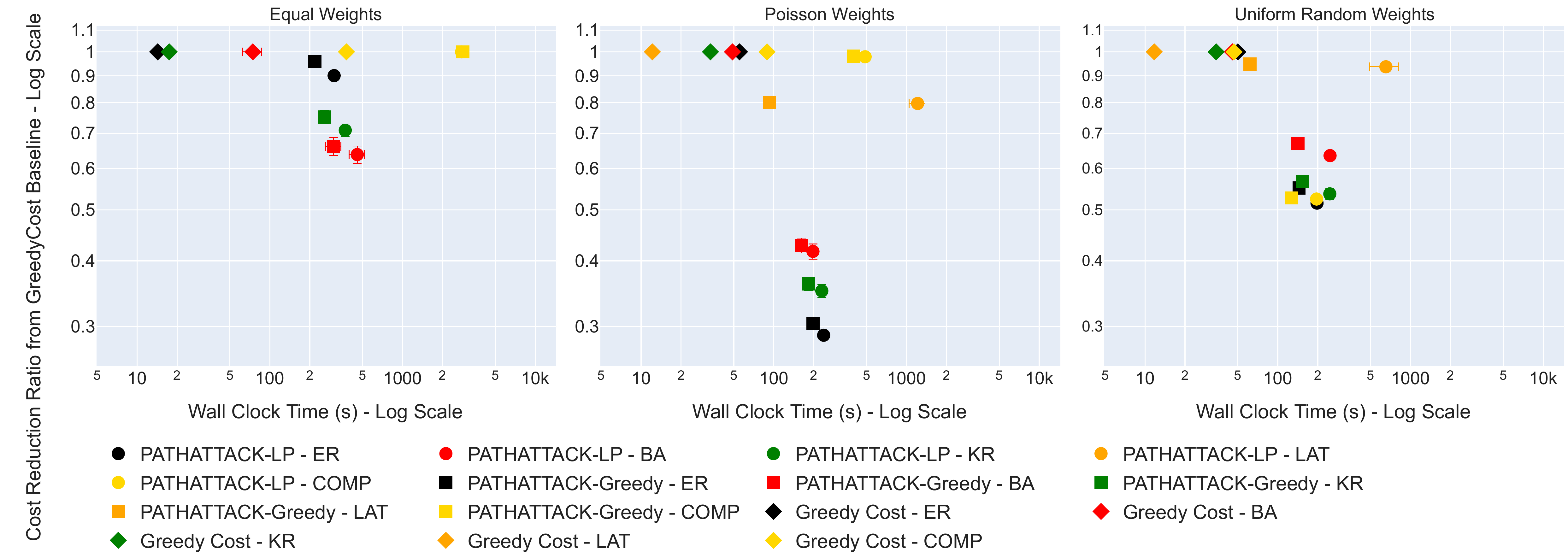}
    \caption{Results on synthetic networks. Shapes represent the different algorithms and colors represent the different networks. The horizontal axis represents wall clock time in seconds and the vertical axis represents the edge removal cost as a proportion of the cost required by the \texttt{GreedyCost} baseline. Lower cost reduction ratio and lower wall clock time is better. \texttt{PATHATTACK} yields a substantial cost reduction for weighted ER, BA, and KR graphs, while the baseline achieves nearly optimal performance for LAT.}
    \label{fig:plots_sim_gs}
\end{figure}

\begin{figure}[t]
    \centering
    \includegraphics[width=\textwidth]{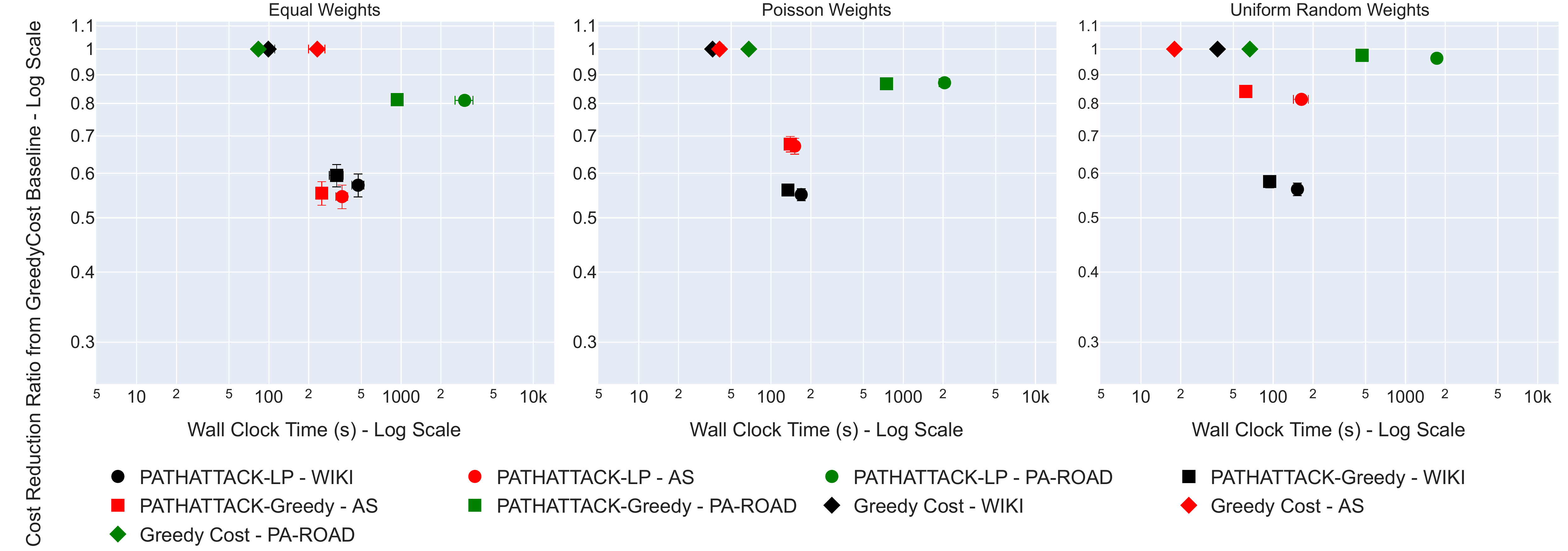}
    \caption{Results on unweighted real networks.
    Shapes represent the different algorithms and colors represent the different networks. The horizontal axis represents wall clock time in seconds and the vertical axis represents the edge removal cost as a proportion of the cost required by the \texttt{GreedyCost} baseline. Lower cost reduction ratio and lower wall clock time is better.
    As with the synthetic networks, \texttt{PATHATTACK} provides a significant cost reduction in networks other than those with grid-like topologies, where the baseline is nearly optimal.}
    \label{fig:plots_ruw_gs}
\end{figure}

\begin{figure}[t]
    \centering
    \includegraphics[width=0.6\textwidth]{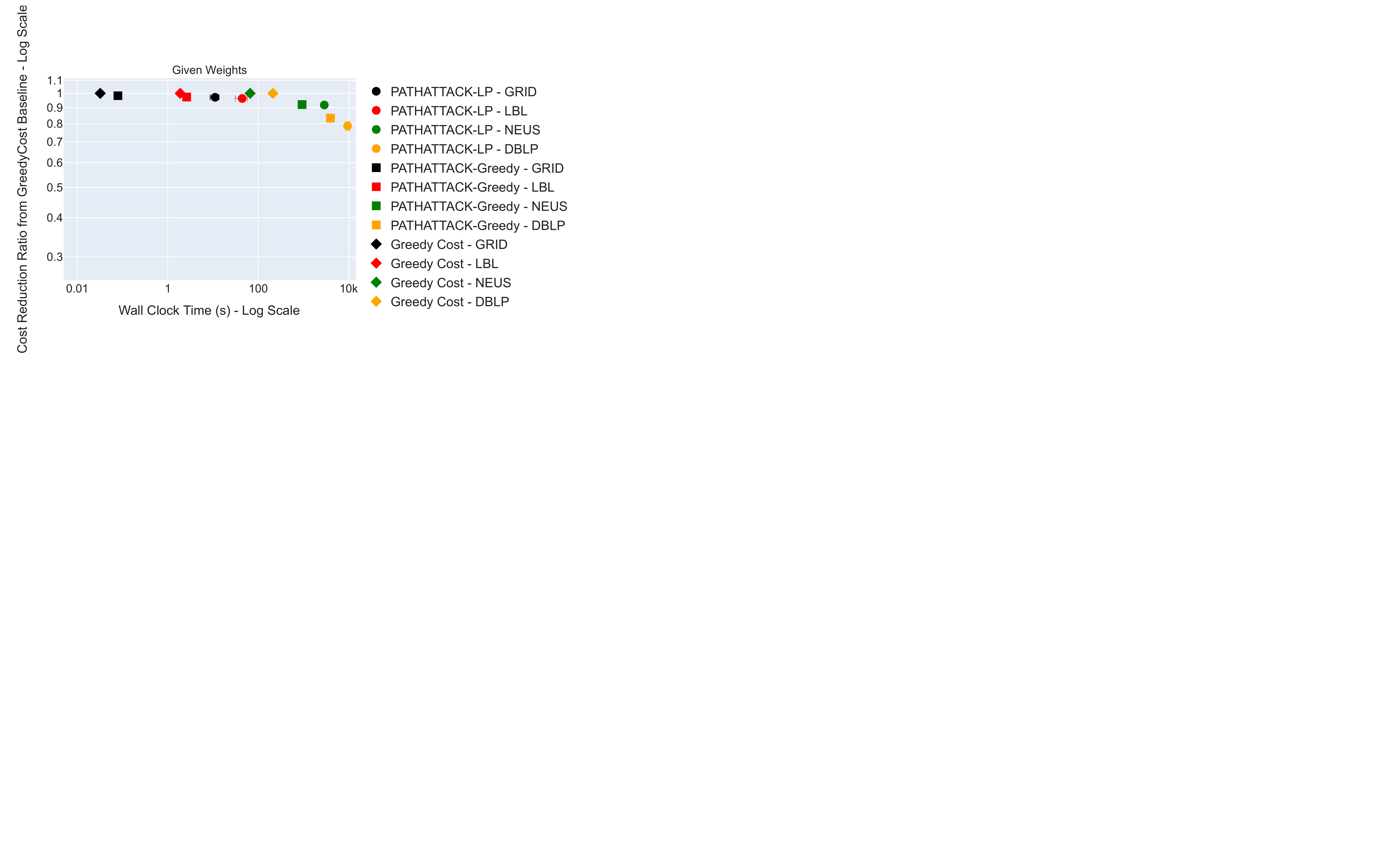}
    \caption{Results on weighted real networks.
    Shapes represent the different algorithms and colors represent the different networks. The horizontal axis represents wall clock time in seconds and the vertical axis represents the edge removal cost as a proportion of the cost required by the \texttt{GreedyCost} baseline. Lower cost reduction ratio and lower wall clock time is better.
    \texttt{PATHATTACK} reduces the cost of attacking the DBLP social network, while the other networks (those with very low clustering) achieve high performance with the baselines. Note: the range of the time axis is lower than that of the previous plots.}
    \label{fig:plots_rw_gs}
\end{figure}

Comparing the cost achieved by \texttt{PATHATTACK} to those obtained by the greedy baselines, we observe some interesting phenomena. Across the synthetic networks in Figure~\ref{fig:plots_sim_gs}, the real graphs with synthetic weights in Figure~\ref{fig:plots_ruw_gs}, and the graphs with real weights in Figure~\ref{fig:plots_ruw_gs}, lattices and road networks have a similar tradeoff: \texttt{PATHATTACK} provides a mild improvement in cost at the expense of an order of magnitude additional processing time. Considering that \texttt{PATHATTACK-LP} typically results in the optimal solution, this means that the baselines are achieving near-optimal cost with an extremely na\"{i}ve algorithm. On the other hand, ER, BA, and KR graphs follow a trend more similar to the AS and WIKI networks, particularly in the randomly weighted cases. Here, the cost is cut by a substantial fraction---enabling an adversary to attack with a smaller budget---for a similar or smaller increase in running time. This suggests that the time/cost tradeoff is much less favorable for less clustered, grid-like networks. 

Cliques (COMP, yellow in Figure~\ref{fig:plots_sim_gs}) are particularly interesting in this case, showing a phase transition as the entropy of the weights increases. When weights are equal for all edges, cliques behave like an extreme version of the road networks: an order of magnitude increase in run time with no decrease in cost. With Poisson weights, \texttt{PATHATTACK} yields a slight improvement in cost, whereas when uniform random weights are used, the clique behaves much more like an ER or BA graph. In the unweighted case, $p^*$ is a three-hop path, so all other two- and three-hop paths from $s$ to $t$ must be cut, which the baseline does efficiently. Adding Poisson weights creates some randomness, but most edges have a weight that is approximately 21, so it is still similar to the unweighted scenario. With uniform random weights, we get the potential for much different behavior (e.g., short paths with many edges) for which the greedy baseline's performance suffers.

There is an opposite, but milder, phenomenon with PA-ROAD and LAT: using higher-entropy weights \emph{narrows} the cost difference between the baseline and \texttt{PATHATTACK}. This may be due to the source and destination being many hops away, as distinct from the clique case. When the terminal nodes many hops apart, many of the shortest paths between them could go through a few very low-weight (thus low-cost) edges. A very low weight edge between two nodes would be very likely to occur on many of the shortest paths, and would be found in an early iteration of the greedy algorithm and removed, while considering more shortest paths at once would yield a similar result. 
We also note that, in the weighted graph data, LBL and GRID behave similarly to the road networks. Among the real datasets we use, these have a very low clustering coefficient (see Appendix~\ref{sec:datasets}). This lack of overlap in nodes' neighborhoods may lead to better relative performance with the baseline, since there may not be a great deal of overlap between candidate paths.

\section{Related Work}
\label{sec:related}

Early work on attacking networks focused on disconnecting them~\cite{Albert2000}. This work demonstrated that targeted removal of high-degree nodes was highly effective against networks with powerlaw degree distributions (such as BA networks), but far less so against random networks. This is due to the prevalence of hubs in networks with powerlaw degree distributions.  Other work has focused on disrupting shortest paths via edge removal, but in a narrower context than ours. Work on the most vital edge problem (e.g., \cite{Nardelli2003}) attempts to efficiently find the single edge whose removal most increases the shortest path between two nodes.  In contrast, we consider a devious adversary that wishes a certain path to become shortest.

There are several other adversarial contexts in which path-finding is highly relevant. Some work is focused on traversing hostile territory, such as surreptitiously planning the path of an unmanned aerial vehicle~\cite{Jun2003}. The complement of this is work on network interdiction, where the goal is to intercept an adversary who is attempting to traverse the graph while remaining hidden. This problem has been studied in a game theoretic context for many years~\cite{Washburn1995}, and has expanded into work on disrupting attacks, with the graph representing an attack plan~\cite{Letchford2013}. In this work, as in ours, oracles can be used to avoid enumerating an exponentially large number of possible strategies~\cite{Jain2011}.

Work on Stackelberg planning~\cite{Speicher2018} is also relevant, though somewhat distinct from our problem. This work adopts a leader-follower paradigm, where rather than forcing the follower to make a specific set of actions, the leader's goal is to make whatever action the follower takes as costly as possible. This could be placed in our context by having the leader (adversary) attempt to make the follower take the longest path possible between the source and the destination, though finding this path would be NP hard in general.

Another related area is the common use of heuristics, such as using Euclidean distances to approximate graph distances~\cite{Rayner2011}. Exploiting deviations in the heuristic  enables an adversary to manipulate automated plans. Fuzzy matching has been used to quickly solve large-scale problems~\cite{Qiao2011}.  Attacks and defenses in this context is an interesting area for inquiry. A problem similar to Stackelberg planning is the adversarial stochastic shortest path problem, where the goal is to maximize reward while traversing over a highly uncertain state space~\cite{Neu2012}. 

The past few years have seen a great deal of work on attacking machine learning methods where graphs are part of the input. Attacks against vertex classification~\cite{Zugner2018,Zugner2019b} and node embeddings~\cite{Bojchevski2019} consider attack vectors that allow manipulation of edges, node attributes, or both in order to affect the outcome of the learning method. In addition, attacks against community detection have been proposed in which a node can create new edges to alter its assignment into a group of nodes by a community detection algorithm~\cite{Kegelmeyer2018}. Our work complements these research efforts, expanding the space of adversarial graph analysis into another important graph mining task.

\section{Conclusions}
\label{sec:conclusion}

We introduce the Force Path Cut problem, in which an adversary's aim is to force a specified path between a source node and a target node to be the shortest by cutting edges within a required budget. Many real-world applications use shortest-path algorithms (e.g., routing problems in computer, power, road, or shipping networks). We show that an adversary can manipulate the network for his strategic advantage. While the Force Path Cut problem is NP-complete, we show how it can be translated into the (Weighted) Set Cover, thus enabling the use of established approximation algorithms to optimize cost within a logarithmic factor of the true optimum. With this insight, we propose the \texttt{PATHATTACK} algorithm, which uses a natural oracle to generate only those constraints that are needed to execute the approximation algorithms. Across various  synthetic and real networks, we find that the \texttt{PATHATTACK-LP} variant identifies the optimal solution in over 98\% of more than 20,000 randomized experiments that we ran. Another variant, \texttt{PATHATTACK-Greedy}, has very similar performance and typically runs significantly faster than \texttt{PATHATTACK-LP}, while a greedy baseline method is faster still with much higher cost.

% ---- Bibliography ----
\bibliographystyle{splncs04}
\bibliography{bibfile.bib}

\appendix
\section{Notation}
Table~\ref{tab:notation} lists the notation used in the paper.
\begin{table}[!ht]
    \centering
    \begin{tabular}{|c|l|}
         \hline
         Symbol& Meaning \\
         \hline
         $G$ & graph\\
         \hline
         $V$ & vertex set \\
         \hline
         $E$ & edge set\\
         \hline
         $N$ & number of vertices\\
         \hline
         $M$ & number of edges\\
         \hline
         $s$ & source vertex\\
         \hline
         $t$ & destination vertex\\
         \hline
         $p^*$ & adversary's desired path in $G$  \\
         \hline
         $w(e)$ & edge weight function $w:E\rightarrow\reals_{\geq0}$\\
         \hline
         $\reals_{\geq 0}$&set of nonnegative real numbers\\
         \hline
         $\mathcal{U}$ & universe for the set cover problem\\
         \hline
         $\mathcal{S}$ & set of subsets of universe $\mathcal{U}$ \\
         \hline
         $\delta_S$ & binary vector representing inclusion of set $S \in \mathcal{S}$ \\
         \hline
         $\wVec$ &edge weight vector\\
         \hline
         $c(e)$ & edge removal cost function $c:E\rightarrow\reals_{\geq0}$\\
         \hline
         $\cVec$ & vector of edge removal costs\\
         \hline
         $b$ & adversary's budget\\
         \hline
         $\Delta$ & binary vector representing edges cut \\
         \hline
         $\xVec$ & binary vector representing edges in a path \\
         \hline
         $P_{p^*}$ & set of paths from $s$ to $t$ no longer than $p^*$\\
         \hline
         $(\cdot)^\top$ & matrix or vector transpose\\
         \hline
         $\langle k\rangle$ & average degree in $G$\\
         \hline 
         $\sigma_k$ & standard deviation of node degrees in $G$ \\
         \hline 
         $\kappa$ & global clustering coefficient in $G$ \\
         \hline
         $\tau$ & transitivity in $G$ \\
         \hline
         $\triangle$ & number of triangles in $G$ \\
         \hline
         $\varphi$ & number of connected components in $G$ \\
         \hline
    \end{tabular}
    \caption{Notation used throughout the paper}
    \label{tab:notation}
\end{table}

\section{Problem Complexity}
If each edge $e$ has a cost of removal $c(e)$, the goal is to minimize the cost of edge removal while forcing the shortest path to be $p^*$. (Cost of removal is not necessarily equal to the weight.) We can show that the $k$-Terminal Cut problem reduces to this one. In the $k$-Terminal Cut problem---described in~\cite{Dahlhaus1994}---we are given a weighted graph $G=(V, E)$ where each edge has a positive weight $w(e)$ and we have $k$ terminal nodes from $V$ and a budget $b$. The goal is to find an edge subset $E^\prime\subset E$ where $\sum_{e\in E^\prime}{w(e)}\leq b$ whose removal disconnects all terminals from one another. This is shown in~\cite{Dahlhaus1994} to be NP-complete for $k>2$. We will show that the solution to Force Path Cut would solve 3-Terminal Cut, thus proving the following theorem.
\begin{theorem}
Force Path Cut is NP-complete for undirected graphs.
\end{theorem}

To solve 3-Terminal Cut, we are given the graph $G=(V, E)$, with weights $w(e)>0$ for all $e\in E$, and a budget $b>0$, along with three terminal nodes $s_1, s_2, s_3\in V$ (following notation from~\cite{Dahlhaus1994}). We will create a new graph $\hat{G}$ from $G$ and use it as an input for Force Path Cut. Let $\wall$ be the sum of all edge weights, i.e., $\wall=\sum_{e\in E}{w(e)}$. Note that if there are any edges between $s_1$, $s_2$, and $s_3$, these edges must be removed in the solution (any solution that does not remove them does not isolate the terminals from one another). Let $E^\prime_t$ be the set of any such edges, i.e.,
\begin{equation}
E^\prime_t=E\cap \left\{\{s_1, s_2\}, \{s_2, s_3\},  \{s_1, s_3\}\right\}.\label{eq:interterminalEdges}
\end{equation}
Since any edges between the terminals must be removed, the budget for removing the remaining edges must be reduced, yielding a new budget 
\begin{equation}
    \hat{b}=b-\sum_{e\in E^\prime_t}{w(e)}.
\end{equation}
Finally, we add edges between the terminal nodes with specific weights: an edge between $s_1$ and $s_2$ and one between $s_2$ and $s_3$, each with weight $\wall+2\epsilon$ for some $\epsilon > 0$, and an edge between $s_1$ and $s_3$ with weight $2\wall+3\epsilon$. Let $\hat{E}$ be the edge set with the new edges added and $\hat{w}$ be the new set of edge weights (with all other weights retained from the original graph). The new graph $\hat{G}=\left(V, \hat{E}\right)$ will be an input to Force Path Cut, with budget $\hat{b}$, starting node $s=s_1$, destination node $t=s_3$, and target edge $p^*=(s_1, s_3)$. Pseudocode for this procedure is provided in Algorithm~\ref{alg:createInput}.
\begin{algorithm}
\begin{algorithmic}
\STATE \textbf{Input:} Graph $G=(V, E)$, weights $w$, budget $b\geq0$, terminals $s_1, s_2, s_3\in V$
\STATE \textbf{Output:} Graph $\hat{G}$, weights $\hat{w}$, unused edges $E^\prime_t$
\STATE $\wall\gets\sum_{e\in E}{w(e)}$
\STATE $E^\prime_t\gets E\cap \left\{\{s_1, s_2\}, \{s_2, s_3\},  \{s_1, s_3\}\right\}$
\STATE $\hat{E}\gets E\setminus E^\prime_t$
\FORALL{$e\in\hat{E}$}
\STATE $\hat{w}(e)\gets w(e)$
\ENDFOR
\STATE $e_{12}\gets \{s_1, s_2\}$\ \ \ $\langle\langle$create new edges$\rangle\rangle$
\STATE $\hat{w}(e_{12})\gets \wall+2\epsilon$
\STATE $e_{23}\gets \{s_2, s_3\}$
\STATE $\hat{w}(e_{12})\gets \wall+2\epsilon$
\STATE $e_{13}\gets \{s_1, s_3\}$
\STATE $\hat{w}(e_{12})\gets 2\wall+3\epsilon$
\STATE $\hat{E}\gets\hat{E}\cup\{e_{12}, e_{23}, e_{13}\}$
\RETURN $\hat{G}=\left(V, \hat{E}\right), \hat{w}, E^\prime_t$
\end{algorithmic}
\caption{Create Force Path Cut input graph}
\label{alg:createInput}
\end{algorithm}

\begin{lemma}
Let $G=(V, E)$ be an undirected graph. For any node subset $V_s\subset V$, if $E$ can be partitioned $E=E_s\cup E_{\bar{s}}$, $E_s\cap E_{\bar{s}}=\emptyset$, such that (1) all edges between nodes in $V_s$ are in $E_s$ and (2) there is no path between any two nodes in $V_s$ within $E_{\bar{s}}$, then all simple paths between nodes in $V_s$ use only edges in $E_s$.\label{lem:edgePartitionPath}
\end{lemma}
\begin{proof}
Suppose a simple path existed between two nodes $u,v\in V_s$ that included edges in $E_{\bar{s}}$. Let $d$ be the number of edges in this path and let $e_i\in E$ for $1\leq i\leq d$ be the sequence of edges starting from $u$ and ending at $v$. Let $j$ be the sequential index of the first edge in $E_{\bar{s}}$ that appears in the path. This edge goes from a vertex in $V_s$ to a vertex in $V\setminus V_s$. (Any edges occurring beforehand are in $E_s$, so only connect nodes within $V_s$, and $e_j$ is in $E_{\bar{s}}$, so at least one vertex is outside of $V_s$.) Let $e_j=\{u_j, v_j\}$, where $u_j\in V_s$ and $v_j\in V\setminus V_s$. Finally, let $e_k$ be the first edge after $e_j$ that connects a node from $V\setminus V_s$ to a node from $V_s$, i.e., the minimum $k > j$ such that $e_k=\{u_k, v_k\}$ where $v_1\in V_s$ or $v_2\in V_s$. We note that such an edge must exist, since the final node in the sequence is in $V_s$. Note also that $e_k\in E_{\bar{s}}$, since one of its vertices is in $V\setminus V_s$. Without loss of generality, let $u_k\in V\setminus V_s$ and $v_k\in V_s$. The edges in the sequence $e_i$ for $j\leq i \leq k$ form a path from $u_j\in V_s$ to $v_k\in V_s$ using only edges in $E_{\bar{s}}$. By the assumption of the lemma there is no path between any two nodes in $V_s$ using edges in $E_{\bar{s}}$, this means that $u_j$ and $v_k$ must be the same node, which contradicts the premise that the path is simple. This completes the proof.
\end{proof}

\begin{algorithm}
\begin{algorithmic}
\STATE \textbf{Input:} Graph $G=(V, E)$, weights $w$, budget $b\geq0$, terminals $s_1, s_2, s_3\in V$
\STATE \textbf{Output:} Boolean value indicating whether the 3 terminals can be separated
\STATE $\hat{G}=\left(V, \hat{E}\right), \hat{w}, E^\prime_t\gets$ output of Algorithm~\ref{alg:createInput}
\STATE $\hat{b}\gets b-\sum_{e\in E^\prime_t}{w(e)}$
\IF{$\hat{b}<0$}
\RETURN \FALSE\ \ \ \ $\langle\langle$budget is too small$\rangle\rangle$
\ENDIF
\STATE $s\gets s_1$
\STATE $t\gets s_3$
\STATE $p^*\gets(s_1, s_3)$
\STATE \textbf{return} ForcePathCut$\left(\hat{G}, \hat{w}, \hat{w}, \hat{b}, p^*, s,t\right)$
\end{algorithmic}
\caption{Solve 3-Terminal Cut via Force Path Cut}
\label{alg:reduction}
\end{algorithm}

\begin{lemma}
Existence of a solution to 3-Terminal Cut implies existence of a solution to Force Path Cut in the graph modified by Algorithm~\ref{alg:createInput}.\label{lem:3cutToForceEdge}
\end{lemma}
\begin{proof}
Let $E^\prime\subset E$ be a solution to 3-Terminal Cut, i.e., a set of edges such that $\sum_{e\in E^\prime}{w(e)}\leq b$ and in the graph $G^\prime=(V, E\setminus E^\prime)$ there is no path connecting any of the terminal nodes $s_1$, $s_2$, and $s_3$. Any edges from $G$ that directly connect the terminals must be in $E^\prime$ or such a path would exist. Letting $E^\prime_t$ be the set of any such edges (as in (\ref{eq:interterminalEdges})), this means that $E^\prime_t\subset E^\prime$.

The graph as modified by Algorithm~\ref{alg:createInput} includes edges $$\hat{E}=\left(E\setminus E^\prime_t\right)\cup\{e_{12},e_{13},e_{23}\},$$ with weights $w(e_{12})=\wall+2\epsilon$, $w(e_{23})=\wall+2\epsilon$, and $w(e_{13})=2\wall+3\epsilon$. By the assumption of the lemma, removing all edges in $E^\prime$ disconnects the terminals from one another. Algorithm~\ref{alg:createInput} starts by removing $E^\prime_t$ and adds 3 new edges to the graph, resulting in the graph $\hat{G}$. Consider a partition of  $\hat{E}$---the edges in $\hat{G}$---into 3 subsets:
\begin{align}
    E_1&=\{e_{12}, e_{23}, e_{13}\}\label{eq:E1}\\
    E_2&=(E^\prime\setminus E^\prime_t)\\
    E_3&=E\setminus E^\prime.\label{eq:E3}
\end{align}
Note that $\hat{E}=E_1\cup E_2\cup E_3$ and
\begin{equation}
    E_1\cap E_2=E_2\cap E_3=E_1\cap E_3=\emptyset,
\end{equation}
so (\ref{eq:E1})--(\ref{eq:E3}) describe a proper partition. Suppose the Force Path Cut procedure removes the edges in $E_2$ from $\hat{G}$. Note that 
\begin{align}
    \sum_{e\in E^\prime}{w(e)}\leq b&\Rightarrow\sum_{e_1\in E^\prime_t}{w(e)}+\sum_{e_2\in E^\prime\setminus E^\prime_t}{w(e)}\leq b\\
    &\Rightarrow\sum_{e_2\in E^\prime\setminus E^\prime_t}{w(e)}\leq b-\sum_{e_1\in E^\prime_t}{w(e)}\\
    &\Rightarrow\sum_{e_2\in E_2}{\hat{w}(e)}\leq \hat{b},
\end{align}
with $\hat{b}$ as defined in Algorithm~\ref{alg:reduction}. This implies that the edges in $E_2$ would be within the budget allocated to Force Path Cut.

After removing the edges in $E_2$, the remaining edges in $\hat{G}$ would be $E_1\cup E_3$. The set  $E_1$ only includes edges among terminals and, by the assumption of the lemma, $E_3$ does not include any path between any two terminal nodes. Thus, by Lemma~\ref{lem:edgePartitionPath}, any path from $s_1$ to $s_3$ after removing $E_2$ includes only edges in $E_1$.

There are therefore two possible paths from $s_1$ to $s_3$: $s_1\rightarrow s_3$ and $s_1\rightarrow s_2\rightarrow s_3$. The latter path has weight
\begin{equation}
    \hat{w}(e_{12})+\hat{w}(e_{23})=2(\wall+2\epsilon)>2\wall+3\epsilon=\hat{w}(e_{13}),
\end{equation}
and thus the former is the shortest path from $s_1$ to $s_3$ in $\hat{G}$. Since $p^*=(s_1, s_3)$, the shortest path from $s=s_1$ to $t=s_3$ is $p^*$, meaning that if there is a solution to 3-Terminal Cut in $G$, there is a solution to Force Path Cut in $\hat{G}$. 
\end{proof}

\begin{lemma}
A solution to Force Path Cut in the graph modified by Algorithm~\ref{alg:createInput} implies a solution to 3-Terminal Cut in the original graph.\label{lem:ForcePathTo3Cut}
\end{lemma}
\begin{proof}
Given a graph $G=(V, E)$, weights $w$, a budget $b$, and terminals $$s_1, s_2, s_3\in V,$$ use Algorithm~\ref{alg:createInput} to compute $\hat{G}$, $\hat{w}$, and $E^\prime_t$. As in Algorithm~\ref{alg:reduction}, compute $\hat{b}$, set $s$, $t$, and $p^*$ and solve Force Path Cut. Let $\hat{E}^\prime$ be the edges that are cut when solving the problem, meaning (1) the shortest path from $s=s_1$ to $t=s_3$ is $p^*$  and (2) $\sum_{e\in\hat{E}^\prime}{\hat{w}(e)}\leq\hat{b}$.

After removing the edges, consider a partition of the remaining edge set into $\{e_{12}, e_{23}, e_{13}\}$  and its complement
\begin{equation}
    \hat{E}_{\bar{t}}=(\hat{E}\setminus\hat{E}^\prime)\setminus \{e_{12}, e_{23}, e_{13}\}.
\end{equation}
Within $\hat{E}_{\bar{t}}$, there is no path between any two terminal nodes. If there were any such path, it would have length at most $\wall$, since 
\begin{equation}
    \sum_{e\in\hat{E}_{\bar{t}}}{\hat{w}(e)}\leq \wall.
\end{equation}
Existence of such a path would have at least one of the following implications:
\begin{itemize}
    \item If such a path $p$ existed between $s_1$ and $s_2$, then $p$ followed by $e_{23}$ would be a path from $s_1$ to $s_3$ with length at most
    \begin{equation}
        \wall+\hat{w}(e_{23})=2\wall+2\epsilon<\hat{w}(e_{13}).
    \end{equation}
    This implies that $s_1\rightarrow s_3$ is not the shortest path from $s_1$ to $s_3$, thus contradicting the assumption of the lemma.
    \item The analogous case for a path from $s_2$ to $s_3$ yields an analogous contradiction.
    \item If such a path existed between $s_1$ and $s_3$, its length would be at most $\wall<\hat{w}(e_{13})$, again contradicting the assumption that $s_1\rightarrow s_3$ is the shortest path from $s_1$ to $s_3$ in $\hat{G}^\prime$.
\end{itemize}
Thus, if we remove $\hat{E}^\prime$ from $\hat{G}$, the only edges connecting the terminal nodes are $e_{12}$, $e_{23}$, and $e_{13}$. In other words, removing $\hat{E}^\prime$ and $\{e_{12}, e_{23}, e_{13}\}$ from $\hat{E}$ will result in the terminals being disconnected from one another. Note that $E$ and $\hat{E}$ are the same after removing nodes between the terminals, i.e.,
\begin{equation}
    E\setminus E_t^\prime = \hat{E}\setminus\{e_1, e_2, e_3\}
\end{equation}
This means that
\begin{align}
    (\hat{E}\setminus\hat{E}^\prime)\setminus\{e_1, e_2, e_3\}&=(\hat{E}\setminus\{e_1, e_2, e_3\})\setminus\hat{E}^\prime\\
    &=(E\setminus E_t^\prime)\setminus \hat{E}^\prime\\
    &=E\setminus(E_t^\prime\cup \hat{E}^\prime).
\end{align}
Thus, removing $E_t^\prime$ and $\hat{E}^\prime$ from the original graph results in a graph with no path between any two terminals. Recall that the assumption of the lemma requires:
\begin{align}
    &\sum_{e\in\hat{E}^\prime}{\hat{w}(e)}\leq\hat{b}=b-\sum_{e\in E_t^\prime}{w(e)}\\
    \Rightarrow &\sum_{e_1\in E_t^\prime}{w(e_1)}+\sum_{e_2\in\hat{E}^\prime}{\hat{w}(e_2)}\leq b\\
    \Rightarrow&\sum_{e\in E_t^\prime\cup\hat{E}^\prime}{w(e)}\leq b.
\end{align}
Removing these edges is, therefore, within the budget allocated. 
\end{proof}

We have now proven that the Theorem Algorithm~\ref{alg:reduction} is a polynomial-time reduction from 3-Terminal Cut to Force Path Cut, as Lemmas~\ref{lem:3cutToForceEdge} and~\ref{lem:ForcePathTo3Cut} have shown. Since 3-Terminal Cut is NP-complete, this implies Force Path Cut is NP-complete as well.

\section{Datasets}
\label{sec:datasets}
Our experiments were run on several synthetic and real networks across different edge-weight initialization. All networks are undirected. We described the edge-weight initialization schemes in Section~\ref{sec:graphdata}. Table \ref{table:syn_net_prop} provides summary statistics of the synthetic networks.

\begin{table}[!ht]
\renewcommand{\arraystretch}{1.0}
\centering
\begin{tabular}{|l|c|c|c|c|c|c|c|c|} 
\hline
Networks & Nodes & Edges & $\langle k\rangle$ & $\sigma_k$ & $\kappa$ & $\tau$ & $\triangle$ & $\varphi$ \\
\hline \hline
ER  & $16,000$ & $159,880$ & $19.985$ & $4.469$ & $0.001$ & $0.001$ & $1,326$ & $1$ \\
 & $\pm 0$ & $\pm38$ & $\pm0.05$ & $\pm0.02$ & $\pm0.0$ & $\pm0.0$ & $\pm39$ & $\pm0$ \\
\hline
BA  & $16,000$ & $159,900$ & $19.987$ & $24.475$ & $0.007$ & $0.006$ & $17,133$ & $1$ \\
 & $\pm0$ & $\pm0$ & $\pm0$ & $\pm0.3$ & $\pm0.0$ & $\pm0$ & $\pm500$ & $\pm0$ \\
\hline
KR  & $16,337$ & $159,595$ & ~$19.537$~ & ~$16.537$~ & ~$0.003$~ & $0.005$ & $8,492$ & ~$1.18$~ \\
& $\pm22$ & $\pm94$ & $\pm0.02$ & $\pm1.32$ & $\pm0$ & ~$\pm0.002$~ & $\pm2,234$ & ~$\pm0.38$~ \\ 
\hline
LAT  & ~$81,225$~ & $161,880$ & $3.985$ & $0.118$ & $0$ & $0$ & $0$ & $1$ \\
 & $\pm0$& $\pm0$& $\pm0$& $\pm0$& $\pm0$& $\pm0$& $\pm0$& $\pm0$ \\ 
\hline
COMP~  & $565$ & ~$159,330$~ & $564$ & $0$ & $1$ & $1$ & ~$29,900,930$~ & $1$ \\
& $\pm0$& $\pm0$& $\pm0$& $\pm0$& $\pm0$& $\pm0$& $\pm0$& $\pm0$ \\
\hline
\end{tabular}
\caption{Properties of the synthetic networks used in our experiments. Our random graph models are: ER (Erd\"{o}s-R\'{e}nyi), BA (Barab\'{a}si-Albert), KR (Stochastic Kronecker), LAT (lattice), and COMP (complete, a.k.a.~clique). For each random graph model, we generate 100 networks. Note that the number of edges across the different networks is $\approx160$K. The table shows the average degree ($\langle k\rangle$), standard deviation of the degree ($\sigma_k$), global clustering coefficient ($\kappa$), transitivity ($\tau$), number of triangles ($\triangle$), and the number of components ($\varphi$). The $\pm$ values show the standard deviation across 100 runs of each random graph model.}
\label{table:syn_net_prop}
\end{table}

We ran experiments on both weighted and unweighted real networks. In cases of unweighted networks, we added either Poisson, uniformly distributed, or equal weights as was the case of synthetic networks. Below, is a brief description of each network used. Table~\ref{table:net_prop} summarizes the properties of each network. 

The unweighted networks are: \textbf{Wikispeedia graph (WIKI)}: The network consists of Web pages (nodes) and connections (edges) created from the user-generated paths in the Wikispeedia game~\cite{West2009}. Available at \url{https://stanford.io/3cLKDb7}. \textbf{Oregon autonomous system network (AS)}: Nodes represent autonomous systems of routers and edges denote communication between the systems~\cite{Leskovec2005}. The dataset was collected at the University of Oregon on 31 March 2001. Available at \url{https://stanford.io/3rLItfN}. \textbf{Pennsylvania road network (PA-ROAD)}: Nodes are intersections in Pennsylvania, connected by edges representing roads~\cite{Leskovec2009}. Available at \url{https://stanford.io/31Jnb7W}.

\begin{table}[ht]
\centering
\begin{tabular}{|l|c|c|c|c|c|c|c|c|} 
\hline
Networks & Nodes & Edges &  $\langle k\rangle$ & $\sigma_k$ & $\kappa$ & $\tau$ & $\triangle$ & $\varphi$ \\
\hline \hline
GRID & 347 & 444 & 2.559 & 1.967 & 0.086 & 0.087 & 40 & 1 \\ 
\hline
LBL & 3186 & 9486 & 5.954 & 25.515 & 0.099 & 0.005 & 1821 & 10 \\
\hline
WIKI & 4,592 & 106,647 & ~46.449~ & 69.878 & 0.274 & 0.102 & 550,545 & 2 \\
\hline
AS & 10,670 & 22,002 & 4.124 & 31.986 & 0.296 & 0.009 & 17,144 & 1 \\
\hline
PA-ROAD~ & ~1,088,092~ & ~1,541,898~ & 2.834 & 1.016 & 0.046 & 0.059 & 67,150 & 206 \\
\hline
NEUS & 1,524,453 & 1,934,010 & 2.537 & 0.950  & 0.022 & 0.030 & 37,012 & 1 \\
\hline
DBLP & 1,836,596 & 8,309,938 & 9.049 & ~21.381~ & ~0.631~ & ~0.165~ & ~26,912,200~ &  ~60,512~ \\ 
\hline
\end{tabular}
\caption{Properties of the real networks used in our experiments. For each network, we are listing the average degree ($\langle k\rangle$), standard deviation of the degree ($\sigma_k$), global clustering coefficient ($\kappa$), transitivity ($\tau$), number of triangles ($\triangle$), and the number of components ($\varphi$).}
\label{table:net_prop}
\end{table}
\vspace*{-12pt}

The weighted networks are: \textbf{Central Chilean Power Grid (GRID)}: Nodes represent power plants, substations, taps, and junctions in the Chilean power grid. Edges  represent transmission lines, with distances in kilometers~\cite{Kim2018}. The capacity of each line in kilovolts is also provided. Available at \url{https://bit.ly/2OjqCPO}. \textbf{Lawrence Berkeley National Laboratory network data (LBL)}: A graph of computer network traffic, which includes counts of the number of connections between machines over time. Counts are inverted for use as distances. Available at \url{https://bit.ly/2PQbOsr}. \textbf{Northeast US Road Network (NEUS)}: Nodes are intersections in the northeastern part of the United States, interconnected by roads (edges), with weights corresponding to distance in kilometers. Available at \url{https://bit.ly/2QWcug9}. \textbf{DBLP coauthorship graph (DBLP)}: This is a co-authorship network~\cite{Benson2018}. We invert the number of co-authored papers to create a distance (rather than similarity) between the associated authors. Available at \url{https://bit.ly/3fytXFS}.

\end{document}